\documentclass[12pt]{article}

\usepackage{amsmath, amsthm, amssymb}
\usepackage{graphicx}
\usepackage[latin1]{inputenc}
\usepackage{fullpage}
\usepackage{url}

\newtheorem{lemma}{Lemma}
\newtheorem{theo}{Theorem}

\newtheorem{fact}{Fact}

\newtheorem{claim}{Claim}

\newcommand{\E}[1]{\mathbf{E}\left[#1\right]}
\newcommand{\Ep}[1]{\mathbf{E}#1}

\newcommand{\V}[1]{\textrm{Var}[#1]}

\newcommand{\ro}[1]{\textrm{root}(#1)}
\newcommand{\Mf}[2]{M(#1,#2)}
\newcommand{\EMf}[2]{\mathbf{E}M(#1,#2)}

\newcommand{\remove}[1]{}

\newcommand{\prelim}[0]{2}
\newcommand{\pendantFirst}[0]{1}
\newcommand{\tightSampling}[0]{10}
\newcommand{\ESrecCDT}[0]{2}
\newcommand{\mainLemmaSparse}[0]{3}

\begin{document}


\title{The Query-commit Problem}
\author{Marco Molinaro\thanks{Tepper School of Business 
Carnegie Mellon University, Pittsburgh, PA
15213, USA, \mbox{molinaro@cmu.edu}} \and R. Ravi \thanks{Tepper School of Business 
Carnegie Mellon University, Pittsburgh, PA
15213, USA, \mbox{ravi@cmu.edu}}}
\date{}		

\maketitle

\begin{abstract}
	In the \emph{query-commit problem} we are given a graph where edges have distinct probabilities of existing. It is possible to query the edges of the graph, and if the queried edge exists then its endpoints are irrevocably matched. The goal is to find a querying strategy which maximizes the expected size of the matching obtained. This stochastic matching setup is motivated by applications in kidney exchanges and online dating.
		
		In this paper we address the query-commit problem from both theoretical and experimental perspectives. First, we show that a simple class of edges can be queried without compromising the optimality of the strategy. This property is then used to obtain in polynomial time an optimal querying strategy when the input graph is sparse. Next we turn our attentions to the kidney exchange application, focusing on instances modeled over real data from existing exchange programs. We prove that, as the number of nodes grows, almost every instance admits a strategy which matches almost all nodes. This result supports the intuition that more exchanges are possible on a larger pool of patient/donors and gives theoretical justification for unifying the existing exchange programs. Finally, we evaluate experimentally different querying strategies over kidney exchange instances. We show that even very simple heuristics perform fairly well, being within 1.5\% of an optimal \emph{clairvoyant} strategy, that knows in advance the edges in the graph. In such a time-sensitive application, this result motivates the use of \emph{committing} strategies.
\end{abstract}


	\section{Introduction}
	
			The theory of matchings is among one of the most developed parts of graph theory and combinatorics \cite{lovasz}. Matchings can used in a variety of situations, ranging from allocation of workers to workplaces to exchange of kidney among living donors \cite{roth1}. However, the uncertainty present in most applications is not captured by standard models. In order to address this limitation we consider a stochastic variant of matchings.
			
	Before presenting the query-commit problem we describe an application in kidney exchanges which motivates our model. Unfortunately current patients who require a kidney transplant far outnumber the available organs. In the United States alone, more than 84,000 patients were waiting for a kidney in 2010 and 4,268 people in such situation died in 2008 \cite{unos}. However, a distinguished characteristic pertaining to kidney transplants is that they can be carried using the organ of a living donor, usually a relative of the patient. Such operations have great potential to alleviate the long waiting times for transplants. 
	
	One major issue is that many operations cannot be executed due to incompatibility between a patient and its donor. In order to overcome this, it is important to consider 2-way exchanges: Suppose patient $A$ and its willing donor $A'$ are not compatible and the same holds for $B$ and $B'$; however, it is still possible that both patients can receive the required organ via transplants from $A'$ to $B$ and from $B'$ to $A$. This situation can be modeled using a \emph{compatibility graph}, where each node represents a patient/donor pair and edges represent the cross-compatibility between such pairs. The set of transplants that maximize the number of organs received is then given by a maximum matching in this graph \cite{roth1, roth2}. 
	
	However, such a model does not take into account uncertainty in the compatibility graph. In practice, preliminary tests such as blood-type and antigen screening are used to determine only the likelihood of cross-compatibility between pairs. Final compatibility can only be determined using a time-consuming test called \emph{crossmatching}, which involves combining samples of the recipients' and donors' blood to check their reactivity. Furthermore, such a test must be performed close to the surgery date, since even the administration of certain drugs may affect compatibility \cite{crossmatching}. That is, the transplant should be executed as soon as it is detected that two patient/donor pairs are determined to be cross-compatible.
	
	The kidney exchange application motivates the \emph{query-commit problem}, which can be described briefly as follows. We are given a weighted graph $G$ where the weight $p_e$ indicates the probability of existence of edge $e$. In each time step we can query an edge $e$ of $G$ and one of the following happen: with probability $p_e$ (corresponding to the event that $e$ actually exists) its endpoints are irrevocably matched and removed from the graph; with probability $1- p_e$ (corresponding to the event that $e$ does not exist) $e$ is removed from the graph. Notice that at the end of this procedure we obtain a matching in $G$, dependent on both the choices of the queries and the randomness of the edges' existence. In the query-commit problem our goal is to obtain a query strategy that maximizes the expected cardinality of the matching obtained. 
		
		\paragraph{Our results.} 	In this paper we address the query-commit problem from both theoretical and experimental perspectives. First, we show that a simple class of edges can be queried without compromising the optimality of the strategy. This result can be used to simplify the decision making process by reducing the search space. In order to illustrate this, we show that employing this property we can obtain in polynomial time an optimal querying strategy when the input graph is sparse. 
		
		Then we turn our attentions to the kidney exchange application, more specifically on instances for the query-commit problem modeled over real data from existing kidney exchange programs. In this context we are able to prove the following result: as the number of nodes grows, almost every such graph admits a strategy which matches almost all nodes. This result support the intuition that more exchanges are possible on a larger pool of patient/donors \cite{david, segev}. More importantly, it shows the potential gains of merging current kidney exchange programs into a nationwide bank.
		
		Finally, we propose and evaluate experimentally different querying strategies, again focusing on the kidney exchange application. We show that even very simple heuristics perform fairly well. Surprisingly, the best among these strategies are on average within 1.5\% of an optimal \emph{clairvoyant} or \emph{non-committing} strategy that knows in advance the edges in the graph. This indicates that the \emph{committing} constraint is not too stringent in this application. Thus, in such time-sensitive application, this result motivates the use of \emph{committing} strategies.

		\paragraph{Related work.} \cite{chen} recently introduced a generalization of the query-commit problem which contains the extra constraint that a strategy cannot query too many edges incident on the same node. In addition to kidney exchange, the authors also point out the usefulness of this model in the context of online dating. Their main result is that a simple greedy querying strategy is a within a factor of $1/4$ from an optimal strategy. \cite{mestre} use a combination of two strategies in order to obtain an improved approximation factor of $1/3.88$. \cite{viswanath} consider a further extension where edges have values and the goal is to find a strategy that maximize the expected value of the matching obtained. Using an LP-based approach they are able to obtain a strategy which is within a constant factor from an optimal strategy.
		
		We note, however, that in the case of the query-commit problem (i.e. when there is no constraint on the number of edges incident to a node that a strategy can query) every strategy which does not stop before querying all permissible edges is a $1/2$-approximation \cite{chen}. This follows from two easy facts: (i) for every outcome of the randomness from the edges, such a strategy obtains a maximal matching and (ii) every maximal matching is within a factor of $1/2$ from a maximum matching. 
		
		The query-commit problem is similar in nature to other stochastic optimization problems with irrevocable decisions, such as stochastic knapsack \cite{goemansKnapsack} and stochastic packing integer programs \cite{goemansIP}. In both \cite{goemansKnapsack, goemansIP} the authors present approximation algorithms as well as bounds on the benefit of using an adaptive strategy versus a non-adaptive one. 
		
		Different forms of incorporating uncertainty have been also studied \cite{SP} and in particular stochastic versions of classical combinatorial problems have been considered in the literature \cite{covering,paths}. Matching also has its variants which handle uncertainty via a 2-stage model \cite{katriel} or in an online fashion \cite{onlinePD, onlineMatching, goel, KVV, unreliable, vazirani}. The latter line of research has been largely motivated by the increasing importance of Internet advertisement.  
		
		The kidney exchange problem, in its deterministic form, has received a great deal of attention in the past few years \cite{david, roth3, roth1, roth2, segev}. In the previous section we argued that 2-way exchanges can increase the number of organs transplanted, but of course larger chains of exchanges can offer even bigger improvements. Unfortunately, considering larger exchanges makes the problem of finding optimal transplant assignments much harder computationally even if all the edges are know in advance \cite{david}. Nonetheless, \cite{david} present integer programming based algorithms which are able to solve large instances of the problem, on scenarios with up to 10,000 patient/donor pairs. The authors point out, however, the importance of considering other models which take into account the uncertainty in the compatibility graph. Finally, \cite{sandholm,unverDynamic,zenios} address the dynamic aspect of exchange banks, where the pool of patients and donors evolve over time. 
		
	The remainder of the paper is organized as follows. In Section \ref{prelim}	we present a more formal definition of the query-commit problem as well as multiple ways of seeing the process in which matchings are obtained from strategies. Then we prove the important structural property (Lemma \ref{pendantFirst}) which is used to obtain in polynomial time optimal strategies for sparse graphs (Section \ref{theoretical}). Moving to the kidney exchange application, we describe in Section \ref{theoKidney} a model for generating realistic compatibility graphs and prove that, as the number of nodes grows, most of these instances admit strategies which match almost all nodes. In Section \ref{experimental} we address the issue of estimating the value of a strategy as well as computing upper bounds on the optimal solution, and conclude by presenting experimental evaluation of querying strategies. As a final remark, the proofs of all lemmas which are not presented in the text are available in the appendix. 
	
	\section{Preliminaries} \label{prelim}

	We use $V(G)$ and $E(G)$ to denote respectively the set of nodes and edges of a given undirected graph $G$, and use $v(G)$ and $e(G)$ to denote their cardinalities. In addition we use $\mu(G)$ to denote the cardinality of the maximum matching in $G$. We define the neighborhood of a node $u$ as the set $N(u) \doteq \{v \in G : (u,v) \in G\}$ and the (edge) neighborhood of an edge $(u,v)$ as the set $N((u,v)) \doteq \{(x,y) \in G: \{u,v\} \cap \{x, y\} \neq \emptyset\}$. Notice an edge is included in its own neighborhood. We ignore isolated nodes in all subsequent graphs. 
	
	  When $G$ is a rooted tree, $\ro{G}$ denotes its root and $G_u$ is the subtree of $G$ which contains $u$ and all of its descendants. When $G$ is a binary tree, we use $l(u)$ and $r(u)$ to denote respectively the left and right children of a node $u \in G$; we also call $G_{l(u)}$ and $G_{r(u)}$ respectively the \emph{left subtree} and \emph{right subtree} of node $u$. The \emph{height} of a rooted tree is the length (in number of nodes) of the longest path between the root and a leaf.
	
	 Throughout the paper we will be interested in weighted graphs $G = (V, E, p)$ where $p : E \rightarrow (0,1]$ associates nonzero weights to the edges of $G$; we refer to them simply as weighted graphs. A \emph{scenario} or \emph{realization} $\sigma$ of $G$, denoted by $\sigma \sim G$, is a subgraph of $G$ obtained by including each edge $e$ independently with probability $p_e$. Note there can be up to $2^{|E|}$ possible realizations in $G$.

		Now we describe, somewhat informally, the dynamics of querying strategies. Consider a weighted graph $G$, a scenario $\sigma \sim G$ and a querying strategy $S$. We start with an empty matching and $S$ makes its first query for an edge $e$ of $G$. If $e \in \sigma$ then $e$ is added to the current matching and we obtain the residual graph (i.e. the set of permissible edges) $R = G \setminus N(e)$. If $e \notin \sigma$ then $e$ is not added to the matching and the residual graph is $R = G \setminus e$. At this point $S$ queries any other edge of $G$; usually we focus on the case that the new edge belongs to $R$, since edges outside $R$ cannot be added to the matching. The process then continues in the same fashion. We remark that $S$ is oblivious to the scenario $\sigma$ and only uses information from previous queries in order to decide its next query.
				
	  In order to make this process more precise we use decision trees to represent querying strategies. In our context, a \emph{decision tree} $T$ is a binary tree with the following properties: (i) each internal node $x \in T$ corresponds to a query for an edge $a(x) \in G$; (ii) all nodes in the right subtree of $x$ correspond to queries for edges in $G \setminus N(a(x))$; (iii) all nodes in the left subtree of $x$ correspond to queries for edges in $G \setminus a(x)$.
	  It is also useful to associate to each node $x \in T$ the residual graph $G^x$ in the following recursive way: $G^{\ro{T}} = G$, $G^{r(x)} = G^x \setminus N(a(x))$ and $G^{l(x)} = G^x \setminus a(x)$. Figure \ref{fig:exDT} presents an example of a decision tree.
	  
	A decision tree $T$ can be interpreted as a querying strategy as follows: First query the edge $a(\ro{T}) \in G$ associated to the root of $T$; if this query is successful then add $a(\ro{T})$ to the current matching and proceed querying using the right subtree of $\ro{T}$, otherwise just proceed querying using the left subtree of $\ro{T}$. Hence, the execution of $T$ over a scenario $\sigma$ induces a path of nodes $x_1, x_2, \ldots, x_k$ in $T$, where $a(x_1), a(x_2), \ldots, a(x_k)$ is the sequence of edges queried by $T$ and $G^{x_1}, G^{x_2}, \ldots, G^{x_k}$ the sequence of residual graphs. 
	
	Notice that, since $T$ only queries permissible edges, the matching obtained is exactly the set of queried edges which belong to $\sigma$; this matching is denoted by $\Mf{T}{G}(\sigma)$. After unpacking previous definitions, we can also write the random matching $\Mf{T}{G}$ recursively as follows (where $\rho = \ro{T}$ and $e = a(\ro{T})$ to simplify the notation): $\Mf{T}{G} = e \cup \Mf{T_{r(\rho)}}{G \setminus N(e)}$ with probability $p_{e}$ and $\Mf{T}{G} = \Mf{T_{l(\rho)}}{G \setminus e}$ with probability $1 - p_{e}$. Thus, the expected size of $\Mf{T}{G}$ is given by
	\begin{eqnarray} 
		\EMf{T}{G} = p_{e} (1 + \EMf{T_{r(\rho)}}{G \setminus N(e)} + (1 - p_{e}) \EMf{T_{l(\rho)}}{G \setminus e}, \label{ESrec}
	\end{eqnarray}
	where we use $\EMf{T}{G}$ instead of $\E{|M(T,G)|}$ to simplify the notation. 
			
	We remark that every strategy can be represented by a decision tree, so we use $T^S$ to denote a decision tree corresponding to a strategy $S$ and use both terms interchangeably.
	
	Making use of the above definitions, we can formally state the query-commit problem: given a weighted graph $G = (V, E, p)$ with $p : E \rightarrow (0,1]$, we want to find a decision tree $T$ for $G$ that maximizes $\EMf{T}{G}$. The value of an optimal solution is denoted by $OPT(G)$.
		
	We are interested in finding computationally efficient strategies for the query-commit problem. Since decision trees may be already exponentially larger then the input graph, our measure of complexity must allow implicitly defined strategies. We say that a strategy is \emph{polynomial-time computable} if the time used to decide the query in each step is bounded by a polynomial on the description of the input graph; this includes any preprocessing time (i.e. time to construct a decision tree). In Section \ref{optSparse} we present structures similar to decision trees that are useful to describe time-efficient strategies.
		
\begin{figure}
	\centering
		\includegraphics[scale=0.60]{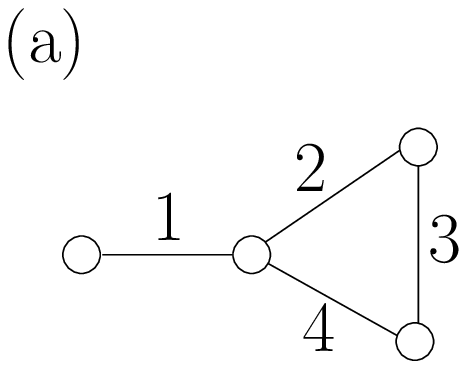} \hspace{90pt}
		\includegraphics[scale=0.60]{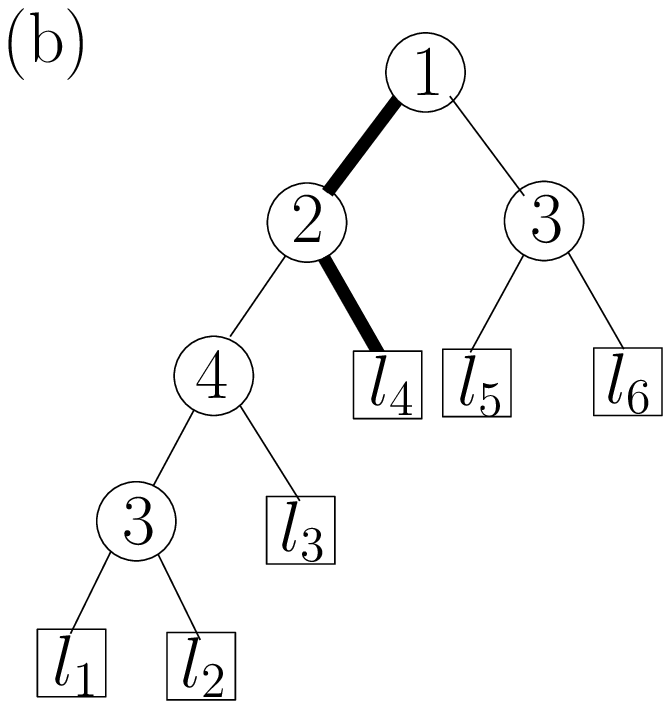}
	\caption{(a) Input graph $G$ with edges labeled from 1 to 4. (b) Example of decision tree $T$ for $G$, with labels corresponding to $a(x)$ for each internal node $x$ of $T$. Let $y$ denote the father of $l_3$, so that $a(y) = 4$; $G^y$ is the subgraph of $G$ consisting of edges $3$ and $4$. The bold path in $T$ is the path obtained by executing $T$ over the scenario $\sigma$ which consists of edges $2$, $3$ and $4$. Moreover, $\Mf{T}{G}(\sigma) = \{2\}$.}
	\label{fig:exDT}
\end{figure}

	
	\section{General theoretical results} \label{theoretical}
	
	We say that an edge of a graph is \emph{pendant} if at least one of its endpoints has degree 1. As a start for our theoretical results, we show that pendant edges can be queried first without compromising the optimality of the strategy. This observation will be fundamental for the development of the polynomial-time computable algorithm for sparse graphs and is also used in the heuristics tested in the experimental section. 
		
	\begin{lemma} \label{pendantFirst}
		Consider a weighted graph $G$ and let $e$ be a pendant edge in $G$. Then there is an optimal strategy whose first query is $e$. 
	\end{lemma}
		

	In order to illustrate the relevance of pendant edges we mention the following result.
	
	\begin{lemma}
		Suppose $G$ is a weighted forest and let $S$ be a strategy that always queries a pendant edge in the residual graph. Then for every $\sigma \sim G$, $|\Mf{S}{G}(\sigma)| = \mu(\sigma)$.
	\end{lemma}


	\subsection{Optimal strategies for sparse graphs} \label{optSparse}
	
	We say that a graph $G$ is $d$-sparse if $e(G) \le v(G) + d$. In this section we exhibit a polynomial-time optimal strategy for $d$-sparse graphs when $d$ is constant. We focus on connected graphs but the result can be extended by considering separately the connected components of the graph. 
	
	So let $G$ be a connected $d$-sparse graph and we further assume (for now) that $G$ does not have any pendant edges; the rationale for the latter is that from Lemma \ref{pendantFirst} we can always start querying pendent edges until we reach a residual graph which has none.
	
	\paragraph{Contracted decision trees.} First, we need to introduce the concept of a \emph{contracted decision tree} (CDT), which generalizes the decision trees introduced in Section \ref{prelim}.

	 Given a strategy $S$ for a weighted graph $H$ we use $\mathcal{R}(S, H)$ to denote the set of possible residual graphs after the execution of $S$, that is, $\mathcal{R}(S, H) = \{ H^x : x \textrm{ is a leaf of } T^S\}$. Then a \emph{contracted decision tree} $T$ for $G$ is a rooted tree where every node $x$ is associated to a residual graph $G^x$ and every internal node $x$ is associated to a query strategy $S^x$ for $G^x$ satisfying the following: (i) $G^{\ro{T}} = G$; (ii) every internal node $x$ of $T$ has exactly $q = |\mathcal{R}(S^x, G^x)|$ children $y_1, y_2, \ldots, y_q$ and $\mathcal{R}(S^x, G^x) = \{G^{y_1}, G^{y_2}, \ldots, G^{y_q}\}$; (iii) if $x$ is an ancestor of $y$ in $T$ then $S^x \neq S^y$. Figure \ref{fig:exCDT} presents an example of a contracted decision tree. 
	 
	\begin{figure}
	\centering
		\includegraphics[scale=0.60]{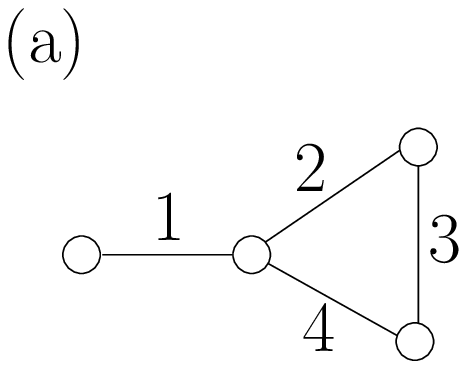} \hspace{45pt}
		\includegraphics[scale=0.60]{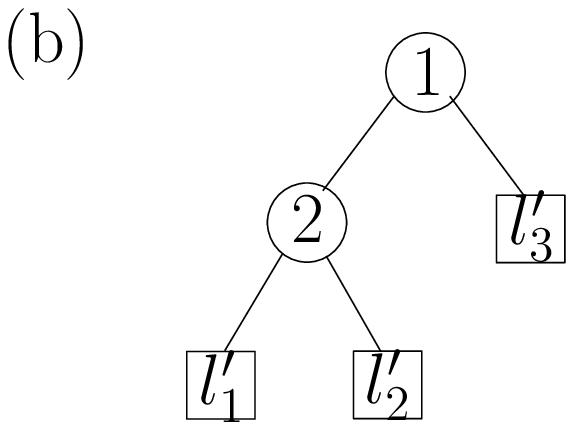} \hspace{45pt}
		\includegraphics[scale=0.60]{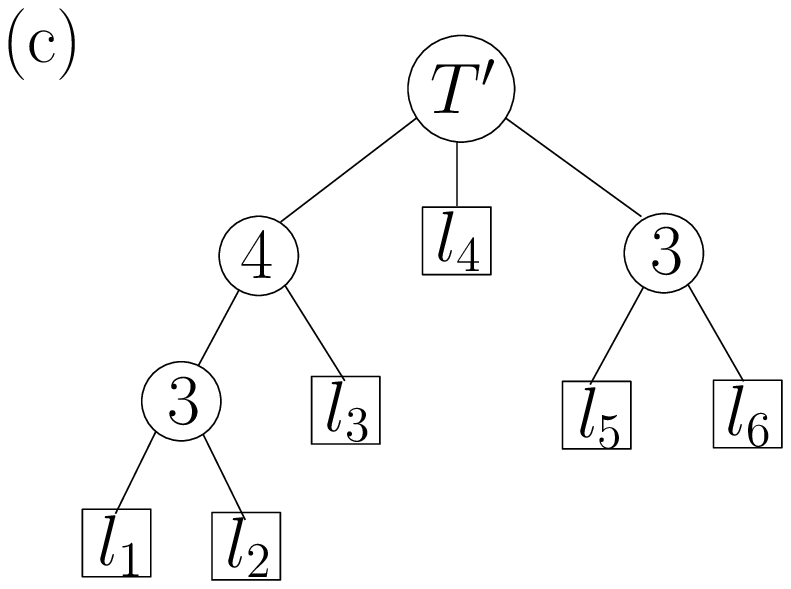}
	\caption{(a) Input graph $G$ with edges labeled from 1 to 4. (b) A decision tree $T'$ for the subgraph of $G$ induced by edges 1 and 2. The set $\mathcal{R}(T', G)$ consists of the graphs $G_1, G_2$ and $G_3$, where $G_1$ is the subgraph of $G$ induced by edges 3 and 4, $G_2$ is the empty graph and $G_3$ is induced by edge 3. (c) A CDT $\bar{T}$ for $G$ where $S^{\ro{\bar{T}}} = T'$. The graphs $G^x$ for the children of $\ro{\bar{T}}$ are, respectively from left to right, $G_1$, $G_2$ and $G_3$. The decision tree corresponding to $\bar{T}$ is $T$ in Figure \ref{fig:exDT}.b. Conversely, $\bar{T}$ can be obtained by contracting $T'$ in $T$.}
	\label{fig:exCDT}
\end{figure}
	 
	 A few remarks are in place. First, condition (iii) is not really fundamental in the definition, although it avoids trivial cases in future proofs. More importantly, notice that a decision tree is simply a CDT where strategy $S^x$ queries a single edge of $G^x$. Also, a CDT can be seen as decision tree where some of its subtrees were contracted into a single node and, conversely, we can obtain a decision tree from a contracted decision by expanding the partial strategies $S^x$'s into decision trees.
	
	A CDT $T$ can be interpreted as a strategy in a similar way as in decision trees: start with an empty matching at the root $\rho = \ro{T}$ and query according to $S^\rho$, which gives a particular residual graph $R \in \mathcal{R}(S^\rho, G^\rho)$ depending on the current scenario $\sigma$; add $\Mf{S^\rho}{G^\rho}(\sigma)$ to the current matching and proceed querying using the subtree $T_x$, where $x$ is the child of $\rho$ with $G^x = R$. The expected size of the matching obtained by a CDT $T$ can be written in an recursive expression similar to \eqref{ESrec}:
	\begin{equation}
		\EMf{T}{G} = \EMf{S^\rho}{G^\rho} + \sum_{x : x \textrm{ is a child of } \rho} \textrm{Pr}(G^x) \cdot \EMf{T_x}{G^x}, \label{ESrecCDT}
	\end{equation}
	where $\textrm{Pr}(G^x)$ is the probability with respect to the scenarios of $G^\rho$ that the residual graph is $G^x$ after employing $S^\rho$ to $G^\rho$.
	
	\paragraph{Decomposition of $G$ and filtering of the strategy space.} Now we turn again to the problem of finding an optimal strategy for $G$. Let $V_{\ge 3}$ be the set of nodes of $G$ which have degree at least 3. Since the nodes in $G \setminus V_{\ge 3}$ have degree at most 2, all of its connected components are either paths of cycles; moreover, we claim that all of them are actually paths. By means of contradiction suppose a connected component of $G \setminus V_{\ge 3}$ is a cycle $C$. Since $G$ is connected, it must contain an edge from a node $u \in C$ to a node in $V_{\ge 3}$. However, this implies that $u$ has degree at least 3 in $G$ and hence $u \in V_{\ge 3}$, contradicting the fact $C$ is a component of $G \setminus V_{\ge 3}$.	
	
	In light of the previous claim it is useful to think about the structure of $G$ is terms of $V_{\ge 3}$ and the paths $P_1, \ldots, P_w$ that are the connected components of $G \setminus V_{\ge 3}$. Notice that the edges of $G$ which have an endpoint in $V_{\ge 3}$ are not present in this decomposition; however, there are at most $6 d$ of these edges, which allows us to ignore them for most part of the discussion of the algorithm. To see this upper bound on the number of such edges first notice that $|V_{\ge 3}| \le 2d$; this holds because $G$ has no pendant edges and hence all of its nodes have degree at least two, so summing over all degrees in $G$ we get	$$2 e(G) \ge 3 |V_{\ge 3}| + 2 (v(G) - |V_{\ge 3}|) \Rightarrow |V_{\ge 3}| \le 2 (e(G) - v(G)) = 2d.$$ Then if $\alpha$ is the number of edges with some endpoint in $V_{\ge 3}$ we again add all degrees of $G$ to obtain $$2 e(G) \ge \alpha + 2 (v(G) - |V_{\ge 3}|) \Rightarrow \alpha \le 2 (e(G) - v(G)) + 2|V_{\ge 3}| \le 6d,$$ obtaining the desired bound.
	
	This result also leads to a bound on the number of paths $P_i$'s. To see this, consider a path $P_i$ and let $u$ and $v$ be its endpoints. The assumption that $G$ does not have pendant edges implies that there are two distinct edges with one endpoint in $\{u,v\}$ and the other in $V_{\ge 3}$. Since there are at most $6d$ of these edges in $G$ and since the $P_i$'s are disjoint, we have that there are at most $3d$ paths $P_i$'s. 
	
	Exploiting the structure of $G$ highlighted in the previous observations, we can construct in polynomial time a CDT which gives an optimal querying strategy. We briefly sketch the argument used to obtain this result. The main observation is that after querying an edge $e$ in $P_i$ we always obtain a residual graph where some edges in $P_i$ are now pendant. Then we can use Lemma \ref{pendantFirst} to keep querying pendant edges, which leads to a residual graph that does not contain any edges of $P_i$. These observations imply that in order to obtain an optimal strategy for $G$ we essentially only need to decide which edge to query first in each $P_i$, as well as an ordering among these edge and the edges with endpoint in $V_{\ge 3}$; all the rest of the strategy follows from querying pendant edges. Moreover, using the fact that there are at most $6d$ edges with endpoint in $V_{\ge 3}$ and at most $3d$ paths $P_i$'s, we can enumerate all these possibilities in time $poly(e(G))$ and obtain the desired result.
	
	Now we formalize these ideas. Consider a subgraph $H$ of $G$ and consider one of the paths $P_i$'s  given by $(u_1 e_1 u_2 e_2 \ldots e_q u_{q+1})$. For an edge $e = e_j \in H$ we define $S(H, e)$ as the strategy which queries edges $e_j, e_{j-1}, \ldots, e_1$ in this order and then queries $e_{j + 1}, e_{j + 2}, \ldots, e_q$ in this order, always ignoring edges which do not belong to $H$. Essentially $S(H, e)$ is querying $e$ first and then edges in $P_i$ which becomes pendant. Notice that there are actually two strategies satisfying the above properties, depending on the orientation of the path $P_i$; so we fix an arbitrary orientation for the paths $P_i$'s in order to avoid ambiguities. 
	
	It follows directly from the definition of $S(H, e)$ that for any residual graph $R \in \mathcal{R}(S(H, e), H)$ we have $R \cap E(P_i) = \emptyset$. More specifically, the set of residual graphs $\mathcal{R}(S(H, e), H)$ can only contain the following graphs: $H \setminus \{u_1, \ldots, u_q\}$, $H \setminus \{u_2, \ldots, u_q\}$, $H \setminus \{u_1, \ldots, u_{q - 1}\}$ and $H \setminus \{u_2, \ldots, u_{q - 1}\}$. For instance suppose nodes $u_1$ and $u_q$ both belong to $H$; then $H \setminus \{u_1, \ldots, u_{q - 1}\}$ is the residual graph of the scenario $\sigma$ iff $u_1$ is the endpoint of edge in $\Mf{S(H, e)}{H}(\sigma)$ and $u_q$ is not. 
	
	Now the next lemma makes formal the assertion that after querying the first edge of $P_i$ we can just proceed by querying pendent edges.
	
	\begin{lemma} \label{mainLemmaSparse}
		Let $H$ be a subgraph of $G$. Let $e$ be an edge in $E(H) \cap E(P_i)$ and suppose that there is an optimal strategy for $H$ that queries $e$ first. Then there is an optimal CDT for $H$ whose root is associated to the strategy $S(H, e)$.
	\end{lemma}		

		Now let $\mathcal{F}$ be the family of all CDT's for $G$ and its subgraphs with the following properties. A CDT $T$ for a subgraph $H$ of $G$ belongs to $\mathcal{F}$ if: (i) $T$ has height at most $9d + 1$; (ii) each node $x$ of $T$ is either associated to a strategy which consists of querying one edge incident to $V(H^x) \cap V_{\ge 3}$ or it is associated to a strategy $S(H^x, e)$ for some $P_i$ and some $e \in E(H^x) \cap E(P_i)$. The usefulness of this definition comes from the fact that we can focus only on this family of CDT's.
		
	  \begin{lemma} \label{Fopt}
			There is an optimal CDT for $G$ which belongs to $\mathcal{F}$.
		\end{lemma}
		
		This lemma is formally proved in the Online Supplement, but the structure of the proof is the following. First we apply Lemma \ref{mainLemmaSparse} repeatedly to show the existence of an optimal CDT satisfying property (ii) of $\mathcal{F}$. To complete the proof, we make use of the fact that there are at most $6d$ edges incident to $V_{\ge 3}$ and at most $3d$ paths $P_i$'s in $G$ to we show that there is one such optimal CDT in $\mathcal{F}$ satisfying the height requirement (i). 
		
		The main point in restricting to CDT's in $\mathcal{F}$ is that there is only a polynomial number of them, as stated in the next lemma.
		
		\begin{lemma} 
			The family $\mathcal{F}$ has at most $e(G)^{4^{2d+1}}$ CDT's. \label{sizeF}
		\end{lemma}
		
		\paragraph{Computing the value of a tree in $\mathcal{F}$.} In light of the previous section, we only need find the best among the (polynomially many) CDT's in $\mathcal{F}$ in order to obtain an optimal strategy for $G$. However, we still need to be able to efficiently calculate the expected size of the matching obtained by each such tree. In this section we show that this can be done recursively employing equation \eqref{ESrecCDT}. 
		
		Consider a CDT $T \in \mathcal{F}$ and let $x$ be a node in $T$. Assume that we have already calculated $\EMf{T_y}{G^y}$ for all proper descendants $y$ of $x$. To calculate $\EMf{T_x}{G^x}$ we consider two different cases depending of $S^x$. 
		
		\medskip \noindent \emph{Case 1: $S^x$ only queries an edge $e$ incident to $V_{\ge 3}$.} As mentioned previously, $\mathcal{R}(S^x, G^x) = \{G^x \setminus N(e), G^x \setminus e\}$. Since $p_e$ is the probability that $e$ belongs to the realization of $G^x$, we have that $\EMf{S^x}{G^x} = p_e$ and the probability that $G^x \setminus N(e)$ is the residual graph is also $p_e$. Therefore, if $y_1$ and $y_2$ are the children of $x$ associated respectively to the residual graphs $G^x \setminus N(e)$ and $G^x \setminus e$, then equation \eqref{ESrecCDT} reduces to:
		\begin{equation*}
			\EMf{T_x}{G^x} = p_e + p_e \EMf{T_{y_1}}{G^{y_1}} + (1 - p_e) \EMf{T_{y_2}}{G^{y_2}}.
		\end{equation*}
		After inductively obtaining all the terms in right hand side of the previous expression, $\EMf{T_x}{G^x}$ can be computed directly.
		
		\medskip \noindent \emph{Case 2: $S^x$ equals to $S(G^x, e)$.} We proceed in the same way as in Case 1, calculating the terms of the right hand side of \eqref{ESrecCDT}. However, these computations are not as straightforward as in the previous case. Given a random matching $M$, let $u \in M$ denote the event that node $u$ is matched in $M$ and let $u \notin M$ be the complementary event. The following lemma is the main tool used during this section and can be obtained via dynamic programming.
		
		\begin{lemma} \label{valueSPi}
			Consider a path $P_i$ and a subgraph $H$ of $G$. Also consider an edge $e \in E(H) \cap E(P_i)$ and define $M = M(S(H, e),H)$. Then there is a procedure which runs in time polynomial in the size of $G$ and computes the values $\mathbf{E}M$, $Pr(u_1 \in M)$, $Pr(u_{q+1} \in M)$ and $Pr(u_1 \in M \wedge u_{q + 1} \in M)$.
		\end{lemma}
		
		The previous lemma directly gives that $\EMf{S^x}{G^x}$ can be computed in polynomial time, so we only need to compute the probabilities that $R \in \mathcal{R}(S^x, G^x)$ is the residual graph after employing $S^x$ to $G^x$. For that, let $P_i = (u_1 e_1 u_2 e_2 \ldots e_q u_{q+1})$ be such that $e \in P_i$. Recall that $\mathcal{R}(S^x, G^x)$ can only contain graphs from the list: $G^x \setminus \{u_1, \ldots, u_{q + 1}\}$, $G^x \setminus \{u_2, \ldots, u_{q + 1}\}$, $G^x \setminus \{u_1, \ldots, u_q\}$ and $G^x \setminus \{u_2, \ldots, u_q\}$. It is easy to see that we can write the probability of obtaining each residual graph in $\mathcal{R}(S^x, G^x)$ using the probabilities in Lemma \ref{valueSPi}. For instance, the probability of obtaining $G^x \setminus \{u_1, \ldots, u_q\}$ is exactly 
		\begin{multline*}
			Pr\left(u_1 \in \Mf{S^x}{G^x} \wedge u_{q + 1} \notin \Mf{S^x}{G^x}\right) \\
			= Pr\left(u_1 \in M(S^x, G^x)\right) - Pr\left(u_1 \in M(S^x, G^x) \wedge u_{q + 1} \in M(S^x, G^x)\right).	
		\end{multline*}
		
		Therefore, Lemma \ref{valueSPi} implies that there is an efficient algorithm to compute all terms in the right hand side of equation \eqref{ESrecCDT}, which gives an efficient way to compute $\EMf{S^x}{G^x}$ as in Case 1. 
		
	\paragraph{Putting everything together.} Consider a connected $d$-sparse graph $G$, possibly containing pendant edges. We define a strategy $S$ for $G$ which proceeds in two steps. First, $S$ queries pendant edges until none exists. At this point we have a residual graph $G'$ with no pendant edges. In the second step it queries according to an optimal strategy $S'$ for $G'$. Applying Lemma \ref{pendantFirst} repeatedly, and using the fact that $S'$ is optimal, we get that $S$ is an optimal strategy for $G$.
	
	In order to prove that $S$ is polynomial-time computable we only need to show that $S'$ is polynomial-time computable. To do so, we need the fact that every connected component of $G'$ is $d$-sparse, which follows from successive applications of the following easy lemma.
	
	\begin{lemma} \label{delSparse}
		Let $G$ be a connected $d$-sparse graph. Then for every edge $e \in E(G)$, the connected components of $G \setminus e$ and $G \setminus N(e)$ are $d$-sparse.
	\end{lemma}

	Let $G'_1, \ldots, G'_k$ be the connected component of $G'$. Since $G'_i$ is $d$-sparse, we can use the tools from previous sections to find an optimal strategies $S'_i$ for the $G'_i$: we enumerate at most $e(G'_i)^{4^{9d + 1}}$ CDT's for $G'_i$ and calculate the value of each of them using the procedure outlined in the previous subsection; then letting $S'_i$ be the strategy among those which has largest value we get that $S'_i$ is optimal for $G'_i$ (cf. Lemma \ref{Fopt}). Then an optimal strategy $S'$ for $G'$ is obtained by querying according to strategy $S'_1$, then $S'_2$ and so on. Notice that the total time needed to compute the $S'_i$'s is bounded by $\sum_{i = 1}^k poly(e(G'_i)^{4^{9d + 1}})$, which is polynomial in $e(G)$ for constant $d$. Thus, $S'$ is polynomial-time computable and we obtain the desired result. 
	
		\begin{theo}		
			Let $G$ be a connected weighted graph satisfying $e(G) \le v(G) + d$ for a fixed $d$. Then there is polynomial-time computable optimal strategy for $G$.
		\end{theo}


	\section{Theoretical results for kidney exchanges} \label{theoKidney}

		Now we focus on the kidney exchange application for the query-commit problem. In this context, weighted graphs are interpreted as weighted compatibility graphs: each node represents a pair patient/donor and the weight of an edge represents the likelihood of cross-compatibility between its endpoints. Our main result in this section is to show that the majority of compatibility graphs admits a simple querying strategy that matches essentially all of its nodes. In light of this result, a strong motivation for creating a large unified bank of patient and donors is obtained.
		
		\subsection{Generating weighted compatibility graphs} \label{kidneyModel}
		
		In order to make the previous claim formal we need to introduce a distribution of compatibility graphs.	In the context of deterministic kidney exchanges, \cite{saidman} introduced a process to randomly generate \emph{unweighted} compatibility graphs. This process is modeled over data maintained by the United Network for Organ Sharing in order to produce realistic instances and several works have considered slight variations of it \cite{david,roth3,roth1,saidman,segev}. Two physiological attributes are considered to determine the incompatibility of a patient and a donor. The first is their ABO blood type, where a patient is blood-type incompatible with a donor if their blood-type pair is one of the following: O/A, O/B, O/AB, A/B, A/AB, B/A and B/AB. The second factor is tissue-type incompatibility and represented by PRA (percent reactive antibody) levels. We now briefly describe the process from \cite{saidman} which generates unweighted graphs and then we mention the slight modification that we use to generate weighted graphs. 
		
		A pair patient/donor is characterized by 5 quantities: the ABO blood type of the patient, the ABO blood type of the donor, the indication if the patient is the wife of the donor, the PRA level of the patient and the indication if the patient is compatible with the donor. A random pair patient/donor is obtained by assigning independently a value for the first four quantities and then picking the compatibility depending on these values. The distribution of these values is described in detail in \cite{saidman} and we only highlight one key property:
		
		\begin{fact}\label{prTypeInc}
			For every pair of blood types $i/j$, the probability that a random pair patient/donor has blood type $i/j$ and is incompatible is nonzero.
		\end{fact}
		
		In order to generate an unweighted compatibility graph on $n$ vertices we first sample random pairs patient/donor until obtaining $n$ incompatible pairs; each pair is added as a vertex to the graph. Then for every pair $u, v$ of nodes in the graph, the probability $p_{(u,v)}$ of cross-compatibility between them is defined based on their physiological characteristics. A key property is that for any two pairs patient/donor $u$ and $v$ the quantity $p_{(u,v)}$ is nonzero if and only if their blood types are cross-compatible (i.e. the patient of $u$ is blood-type compatible with the donor of $v$ and the patient of $v$ is blood-type compatible with the donor of $u$). More specifically, there is a constant $c_1 < 1$ independent of $n$ such that 
		\begin{equation}
			p_{(u,v)} > 0 \Rightarrow p_{(u,v)} \ge 1 - c_1. \label{LBprob}
		\end{equation}
	To conclude the construction of the graph, for each pair of nodes $u,v$ a coin is flipped independently and with probability $p_{(u,v)}$ an edge is added between $u$ and $v$. 
		
		Finally, the modification of the above procedure to generate weighted compatibility graphs consists of changing the last step: if $p_{(u,v)} > 0$ then the edge $(u,v)$ is added to the graph and $p_{(u,v)}$ becomes its weight.
		
		\subsection{An almost optimal strategy}

		In this section we present a simple querying strategy that achieves $\EMf{S}{G} \approx v(G)/2$ for almost all weighted compatibility graphs $G$ generated by the above procedure. To obtain such strategy we decompose the graph into cliques and complete bipartite graphs based on blood-type compatibility. Then we obtain good strategies for these structured subgraphs and finally compose them into a strategy for the original graph.
		
		Let $\mathcal{D}^n$ be the distribution of weighted compatibility graphs on $n$ nodes generated by the procedure from the previous section. For a graph $G$ in $\mathcal{D}^n$ we use $V_{i,j}(G)$ to denote the subset of vertices corresponding to the patient/donor pairs which have blood type $i/j$. The lower case version $v_{i,j}(G)$ is used to denote the cardinality of $V_{i,j}(G)$. 
		
		Consider a random graph $G \sim \mathcal{D}^n$ and a node $u$ of it. The first observation is that the probability that $u$ has patient/donor blood type $i/j$ is equal to $$Pr(\textrm{$P/D$ has blood type $i/j$ $|$ $P$ and $D$ are incompatible}),$$ where $P/D$ is a random pair patient/donor. Using the definition of conditional probability and Fact \ref{prTypeInc} we get that this probability is nonzero. Since we have finitely many blood types, this means that there is a constant $c_2$ independent of $i/j$ and $n$ such that 
		\begin{equation}
			Pr(\textrm{$u$ has patient/donor blood type $i/j$}) \ge c_2 \textrm{ for all } i,j. \label{LPBloodType}
		\end{equation}
		This property, together with the symmetry of blood types of patients and donors, gives the following fact.
		
		\begin{fact} \label{sizeBloodTypes}
			The following properties hold for every $i,j$ (where the expectation is taken with respect to the distribution $\mathcal{D}^n$):
			\begin{enumerate}
				\item $\E{v_{i,j}} = \Omega(n)$
				\item $\E{v_{i,j}} = \E{v_{j,i}}$
			\end{enumerate}
		\end{fact}
		
		In order to describe and analyze the proposed querying strategy, we first focus on the set of graphs in $\mathcal{D}^n$ which have a typical number of nodes associated to each pair of blood types. That is, for $\alpha > 0$ we consider $\mathcal{G}^n_{\alpha}$, which is defined as the set of all graphs $G$ in $\mathcal{D}^n$ which satisfy $$(1 - \alpha) \E{v_{i,j}} \le v_{i,j}(G) \le (1 + \alpha) \E{v_{i,j}}.$$ We first show how to obtain good strategies for graphs in $\mathcal{G}^n_{\alpha}$ and then argue that most of the graphs in $\mathcal{D}^n$ are in this family. 
		
		Fix $\alpha \in (0, 1)$ and consider a graph $G$ in $\mathcal{G}^n_{\alpha}$. Let $G_{i,j}$ be the induced subgraph $G[V_{i,j}(G) \cup V_{j,i}(G)]$. Clearly these subgraphs partition the nodes of $G$. Moreover, every two nodes in $G_{i,j}$ are blood-type cross-compatible, since a patient is blood-type compatible with a donor with the same blood type. Therefore, the construction of $G$ (and more specifically the properties of $p(u,v)$) implies that $G_{i,j}$ is a complete bipartite graph if $i \neq j$ and a complete graph if $i = j$. The fact that $G \in \mathcal{G}^n_{\alpha}$ and part 2 of Fact \ref{sizeBloodTypes} additionally give the following: if $i \neq j$ there is complete bipartite subgraph $G'_{i,j}$ of $G_{i,j}$ with equally sized vertex classes and with $v(G'_{i,j}) \ge 2 (1 - \alpha) \E{v_{i,j}}$. 
		
		The motivation for partitioning $G$ is that there are very simple strategies that work well in complete (bipartite) graphs, as shown in the next two lemmas. 
		
		\begin{lemma} \label{Sbipartite}
			Let $G$ be a weighted complete bipartite graph with $k$ vertices in each vertex class. Then for every $\epsilon > 0$ there is a strategy $S$ such that $\EMf{S}{G} \ge (1 - q^{\lfloor \epsilon k \rfloor} - \epsilon) k$, where $q = \max\{1 - p_e : e \in E(G)\}$. 
		\end{lemma}		
	
	The analogous lemma when $G$ is a clique can be proved by applying the previous lemma to a complete bipartite subgraph of $G$ with vertex classes containing $\lfloor v(G)/2 \rfloor$ vertices. 
	
		\begin{lemma} \label{Scomplete}
			Let $G$ be a weighted complete graph with $k$ vertices. Then for every $\epsilon > 0$ there is a strategy $S$ such that $\EMf{S}{G} \ge (1 - q^{\lfloor \epsilon (k - 1)/2 \rfloor} - \epsilon) \left\lfloor k/2\right\rfloor$, where $q = \max\{1 - p_e : e \in E(G)\}$.  
		\end{lemma}

		Let $S'_{i,j}$ be the strategy given by Lemma \ref{Sbipartite} for $G'_{i,j}$ and $S_i$ be the strategy given by Lemma \ref{Scomplete} for $G_{i,i}$. Since the graphs $G'_{i,j}$'s and $G_{i,i}$'s are disjoint, we can apply the strategies $S'_{i,j}$'s and $S_i$'s sequentially and obtain a strategy $S$ for $G$ such that $$\EMf{S}{G} = \frac{1}{2} \sum_{i \neq j} \EMf{S'_{i,j}}{G'_{i,j}} + \sum_i \EMf{S_i}{G_{i,i}},$$ where the factor $1/2$ in the right hand side appears because the graph $G_{i,j} = G_{j,i}$ is counted twice. Defining $q = \max \{1 - p_e : e \in E(G)\}$ and employing the bounds from Lemmas \ref{Sbipartite} and \ref{Scomplete}, we have that for every $\epsilon > 0$
		\begin{equation*}
			\EMf{S}{G} \ge \frac{1}{2} \sum_{i\neq j} \left(1 - q^{ \left\lfloor \epsilon \frac{v(G'_{i,j})}{2} \right\rfloor} - \epsilon \right) \frac{v(G'_{i,j})}{2} + \sum_i \left(1 - q^{\left\lfloor \epsilon \frac{v(G_{i,i}) - 1}{2} \right\rfloor} - \epsilon \right) \left \lfloor \frac{v(G_{i,i})}{2} \right \rfloor.
		\end{equation*}
	
		Since $G \in \mathcal{G}^n_\alpha$, Fact \ref{sizeBloodTypes} gives that $v(G_{i,i}) = v_{i,j}(G) \ge (1- \alpha) \E{v_{i,j}} = \Omega(n)$ and $v(G'_{i,j}) \ge 2(1 - \alpha) \E{v_{i,j}} = \Omega(n)$. This gives the following asymptotic bound on the quality of strategy $S$ (assuming $n$ large enough):
		\begin{align*}
			\EMf{S}{G} & \ge \left(1 - q^{\epsilon \Omega(n)} - \epsilon \right) \left[ \sum_{i \neq j} \frac{v(G'_{i,j})}{4} + \sum_i \left \lfloor \frac{v(G_{i,i})}{2} \right \rfloor \right] \\
			& \ge \left(1 - q^{\epsilon \Omega(n)} - \epsilon \right) \left[ \sum_{i \neq j} \frac{(1 - \alpha) \E{v_{i,j}}}{2} + \sum_i \frac{(1 - \alpha) \E{v_{i,i}}}{2} - 4 \right] \\
			& = \left(1 - q^{\epsilon \Omega(n)} - \epsilon \right) \left[ \frac{(1 - \alpha) n}{2} - 4 \right] \ge \left(1 - c_1^{\epsilon \Omega(n)} - \epsilon \right) \left[ \frac{(1 - \alpha) n}{2} - 4 \right],
		\end{align*}
		where $c_1$ is satisfies \eqref{LBprob}. Since $c_1 < 1$ it follows that $\lim_{n \rightarrow \infty} \EMf{S}{G} / (n/2) = (1 - \epsilon) (1 - \alpha)$ and thus $S$ matches almost all nodes of the typical graph $G$ for sufficiently large $n$. 
		
		Now we argue that most graphs in $\mathcal{D}^n$ belong to the family $\mathcal{G}^n_{\alpha}$. That is, we want to lower bound the probability that a random graph $G \sim \mathcal{D}^n$ has values $v_{i,j}(G)$'s concentrated around the expectation. 
		
		Consider a random graph $G \sim \mathcal{D}^n$. Since the blood type of each node of $G$ is chosen independently, $v_{i,j}(G)$ is  distributed according to a binomial process of $n$ trials. Therefore, using Hoeffding's inequality \cite{BoucheronLB03} we get that the probability that $v_{i,j}(G)$ lies outside the interval $[(1 - \alpha) \E{v_{i,j}} ,(1 + \alpha) \E{v_{i,j}}]$ is at most $2 \exp(- 2 n^3 (c_2 \alpha)^2)$, where $c_2$ is the constant in \eqref{LPBloodType}. Since there are only 4 blood types, we can employ the union bound to estimating the probability that every $v_{i,j}(G)$ is close to its expected value. With this we obtain that the probability that the generated compatibility graph belongs to $\mathcal{G}^n_{\alpha}$ is at least $1 - 32 \exp(-2 n^3 (c_2 \alpha)^2)$, which goes to 1 exponentially fast with respect to $n$.
		
		Combining the fact that there is a good strategy for graphs in $\mathcal{G}^n_{\alpha}$ with the fact that most graphs of $\mathcal{D}^n$ belongs to $\mathcal{G}^n$ gives the desired result. 
		
		\begin{theo} \label{thmKidney}
			For any $\epsilon > 0$ there is an $n_0(\epsilon)$ such that the following holds for every $n \ge n_0(\epsilon)$. Consider a random compatibility graph $G \sim \mathcal{D}^n$. Then with probability $1 - \epsilon$ over the distribution of $G$ there is a polynomial-time computable strategy $S$ which achieves $\EMf{S}{G} \ge (1 - \epsilon) n/2$. 
		\end{theo}


	\section{Computational results} \label{experimental}
		
		In this section we present an experimental evaluation of the performance of simple querying strategies. During our tests, we decided to focus on the application of the query-commit problem to kidney exchanges and therefore all weighted graphs used in the tests were generated randomly according to the procedure described in Section \ref{kidneyModel}. The results show that practical heuristics perform surprisingly well, \emph{even when compared to optimal non-committting strategies}. As a preparation to our experimental results, we first address the issue of estimating the value of a strategy and estimating an upper bound on $OPT$.
		
		\subsection{Estimating the value of a strategy} \label{value}
		
		 Section \ref{optSparse} already indicated that it is not a trivial task to calculate this value exactly. Given a weighted graph $G$ and a strategy $S$, the direct way of calculating $\EMf{S}{G}$ involves finding $\Mf{S}{G}(\sigma)$ for every scenario $\sigma \sim G$. However, there is often an exponential (in $e(G)$) number of such scenarios. When $S$ is given as a decision tree $T^S$ we can compute $\EMf{S}{G}$ using the recurrence \eqref{ESrec}; but again this procedure takes time proportional to the number of nodes in $T^S$, which can be exponential in $e(G)$ and is therefore also impractical. 
		
		 Despite previous studies on the query-commit and related problems \cite{chen,goemansKnapsack, goemansIP}, there is currently no efficient algorithm to compute the expected value of a strategy and most results rely on an estimation of this value by sampling a subset of the scenarios. Our goal in this section is to address how well we can estimate $\EMf{S}{G}$ by sampling from $G$ and establish the accuracy as function of the number of samples. 
		
		Given a weighted graph $G$ on $n$ nodes and a strategy $S$, the natural way to obtain an unbiased estimation of $\EMf{S}{G}$ is to sample independently $k$ scenarios $\sigma_1, \ldots, \sigma_k \sim G$ and take $\eta = (1/k) \sum_{i = 1}^k |\Mf{S}{G}(\sigma_i)|$ as the estimate. Clearly $\E{\eta} = \EMf{S}{G}$. Moreover, since $|\Mf{S}{G}(\sigma_i)| \in [0, n/2]$ for all $i$ we also have that $\eta \in [0, n/2]$. In order to show that $\eta$ is concentrated around the expectation we can simply employ Hoeffding's inequality to obtain that for every $t \ge 0$
		 \begin{equation}
		 	Pr\left(|\eta - \EMf{S}{G}| \ge t\right) \le 2 \exp\left( -\frac{8 k t^2}{n^2} \right). \label{conc1}
		 \end{equation}
	According to this expression, we need approximately $0.4611 n^2 / t^2$ samples in order to obtain a 95\% confidence interval equal to $\EMf{S}{G} \pm t$. 
	
	Notice that the previous bound does not use much of the structure of the matchings. In particular, \eqref{conc1} relies solely on the fact that the size of the matchings lie in $[0, n/2]$. However, this simple concentration estimate is essentially best possible. This holds because we can construct a graph and a strategy $S$ for it which obtains a small matching half of the time and a large matching half of the time. The variance on the size of the matching obtained by $S$ is then essentially as large as possible when compared to any random variable taking values in $[0, n/2]$; in such case, Hoeffding's inequality is rather tight. More formally, we have that following lemma.
	
	\begin{lemma} \label{tightSampling}
		There is a graph $G$ on $n$ nodes and a strategy $S$ for querying $G$ such that
		\begin{equation*}
			Pr\left(\left|\eta - \EMf{S}{G}\right| \ge t\right) \ge \frac{2}{15} \exp\left( -\frac{64 k t^2}{n^2} \right)
		\end{equation*}
		for all $t \le n/16$.
	\end{lemma}
			
	\subsection{Upper bound on $OPT$} \label{secUB}
		
		Clearly the size of a maximum cardinality matching in $G$ is an upper bound on $OPT(G)$, since the maximum matching in any realization of $G$ has size at most $\mu(G)$.
		 However, this bound can be made arbitrarily weak by considering small edge weights, e.g. $G$ is a set of disjoint edges, each with weight $p$, so $\mu(G) = e(G)$ but $OPT(G) = p \cdot e(G)$. 
		
		Notice that actually $OPT$ can be upper bounded by the expected size of the maximum matching over the realizations of $G$; that is, for $\sigma \sim G$ we have $OPT(G) \le \E{\mu(\sigma)} \doteq \E{\mu}$. This bound is tighter than $\mu(G)$ and, as Section \ref{exp} supports, is oftentimes very close to $OPT(G)$. An important remark is that $\E{\mu}$ is a valid upper even for the \emph{non-commit} or \emph{clairvoyant} version of the problem, where the strategy can first find out exactly the edges in the realization and then decide which matching to take. Again we encounter the issue of calculating or at least estimating $\E{\mu}$. 
		
		\paragraph{Estimating $\E{\mu}$.} Clearly the sample average estimator used in the previous section can also be used to estimate $\E{\mu}$ and the Hoeffding-based bound still holds. However, we can get substantially better concentration results by bounding the variance of the estimate more carefully. We make use of this tighter concentration to reduce the computational effort to estimate the upper bound on $OPT$.
		
		Again consider a weighted graph $G$ on $n$ nodes and $m$ edges and consider $k$ independent scenarios $\sigma_1, \ldots, \sigma_k \sim G$. We take $\eta = (1/k) \sum_{i = 1}^k \mu(\sigma_i)$ as the estimation of $\E{\mu}$, since clearly $\E{\eta} = \E{\mu}$. Our goal now is to bound the variance of $\sigma_i$, which will then be used to provide tighter concentration results for $\eta$. For that we need to introduce the concept of self-bounding functions. 
		
		A nonnegative function $g: \mathcal{X}^d \rightarrow \mathbb{R}$ is \emph{self-bounding} if there exist functions $g_i : \mathcal{X}^{d - 1} \rightarrow \mathbb{R}$ such that for all $x_1, \ldots, x_d \in \mathcal{X}$ the following hold:
		\begin{equation*}
			0 \le g(x_1, \ldots, x_d) - g_i(x_1, \ldots, x_{i-1}, x_{i+1}, \ldots, x_d) \le 1 \ \ \ \textrm{for all } i = 1, \ldots, d
		\end{equation*}
		and
		\begin{equation*}
			\sum_{i = 1}^d \left[g(x_1, \ldots, x_d) - g_i(x_1, \ldots, x_{i-1}, x_{i+1}, \ldots, x_d) \right] \le g(x_1, \ldots, x_d)
		\end{equation*}

		The following lemma motivates the definition of self-bounding functions.
		
		\begin{lemma}[\cite{BoucheronLB03}] \label{selfBoundingVar}
			Suppose $g: \mathcal{X}^d \rightarrow \mathbb{R}$ is a measurable self-bounding function. Let $X_1, \ldots, X_d$ be independent random variables with support on $\mathcal{X}$ and let $Z = g(X_1, \ldots, X_d)$. Then $$\V {Z} \le \E{Z}.$$
		\end{lemma}

		The connection between the previous lemma and our goal of estimating $\V{\mu(\sigma_i)}$ comes from the fact that $\sigma_i$ can be seen as $e(G)$ independent indicator random variables for the edges of $G$ and $\mu(\sigma_i)$ can be see as a self-bounding function. To make this precise let $e_1, e_2, \ldots, e_m$ be the edges of $G$ and let $X_1, X_2, \ldots, X_m$ be independent Bernoulli random variables with $Pr(X_i = 1) = p_{e_i}$. Also, for an indicator vector $x \in \{0,1\}^m$ of the edges of $G$, let $\mu'(x)$ be the size of the maximum matching in the subgraph of $G$ induced by $x$ (i.e. which contains the edge $e_i$ iff $x_i = 1$). It is easy to see that $\mu'(X_1, X_2, \ldots, X_m) = \mu(\sigma_i)$ and the next lemma asserts that $\mu'(.)$ is self-bounding.
		
		\begin{lemma} \label{selfBounding}
			The function $\mu'(.)$ is self-bounding.
		\end{lemma}
		
		Since $\mu(\sigma_i) = \mu'(X_1, \ldots, X_m)$, Lemma \ref{selfBoundingVar} gives the desired bound $\V{\mu(\sigma_i)} \le \E{\mu(\sigma_i)} \le n/2$. Now that we have a handle on this variance, we can evoke Bernstein's inequality \cite{BoucheronLB03} to bound the concentration of $\eta = (1/k) \sum_{i = 1}^k \mu(\sigma_i)$ as follows:
	\begin{equation} 
		Pr\left(\left|\eta - \E{\mu}\right| \ge t\right) \le 2 \exp\left( -\frac{k t^2}{n + 2t/3} \right). \label{boundEmu}
	\end{equation}
	
	When compared to \eqref{conc1}, the lack of an extra term $1/n$ in the exponent makes \eqref{boundEmu} a much stronger bound.
						 		

	\subsection{Comparison of heuristics} \label{exp}
		
		In this section we present experimental results comparing simple querying strategies over random weighted compatibility graphs. All strategies use to some extent an optimization based on Lemma \ref{pendantFirst}, that is, they query pendant edges first if one exists. The heuristics considered in the experiments are described next.
		
		\begin{description}
			\item[Maximum probability.] This strategy first queries pendant edges in the residual graph. In case none exists, it queries the edge with highest weight. 
			
			\item[Minimum probability.] Similar to the previous strategy but edges are queried by decreasing weights.
			
			\item[Minimum degree.] This strategy first queries pendant edges in the residual graph. In case none exists, it queries the edge which has minimum degree in the residual graph, where the degree of an edge is defined as the sum of the degrees of its endpoints.
			
			\item[Minimum average degree.] The average degree of an edge $(u,v)$ is defined as the sum of the weights of edges incident to $u$ plus the sum of the weights of edges incident to $v$. This strategy first queries pendant edges in the residual graph. In case none exists, it queries the edge which has minimum average degree in the residual graph.  
			
			\item[Batch successive matching.] First, this strategy queries pendant edges in the residual graph. After no more pendant edges exist, it finds a maximum cardinality matching in the residual graph. Then it queries all the edges in this matching in an arbitrary order. After all these edges are queried the process is repeated.
			
			\item[Batch successive weighted matching.] Similar to the previous strategy, but now in each round it computes a maximum \emph{weighted} matching with edge weights $1 - p$. 
			
			\item[Successive weighted matching $q$.] First, this strategy queries pendant edges in the residual graph. After no more pendant edges exist, it finds a maximum weighted matching (with edge weights $1 - p$) in the residual graph. Then it queries \emph{one} arbitrary edge in this matching. The process is then repeated.
			
			\item[Successive weighted matching $p$.] Similar to the previous strategy, but the edge weights used in the maximum weighted matchings are simply $p$.
	\end{description}
	
		The instances used in the experiments consist of random weighted compatibility graphs on 100 nodes, generated as described in Section \ref{kidneyModel}. Simulating exchange pools with 100 patient/donor pairs is optimistic but not unrealistic, as pointed out in \cite{roth3}. We also carried experiments in graphs with fewer than 100 nodes, but these graphs did not seem to be large enough to discriminate the querying strategies. 
		
		In order to estimate the value of each strategy we employed the sample average estimation discussed in Section \ref{value}. For each execution, we used 38,000 samples in order to obtain a good estimate; according to inequality \eqref{conc1}, with $0.95$ probability, the estimate is within $\pm 0.35$ of the actual expected matching obtained by the strategy. In order to obtain an estimated upper bound on $OPT$ we used the sample average of $\E{\mu}$ as described in Section \ref{secUB}. The number of samples was chosen to obtain an estimate which is within $\pm 0.1$ of $\E{\mu}$ with probability $0.95$. 
		
		The results of the experiments are presented in Table \ref{tab:heu}. The first eight columns correspond to the strategies in the same order as they were described and the last column presents the upper bound $\E{\mu}$ on $OPT$. In each row, except the last one, we present the estimated value of the strategies on a given instance. The last row of the table indicates the average value of each strategy over all instances. 		
		
\begin{table}
	\begin{center}
	\begin{tabular}{|c|c|c|c|c|c|c|c|c|}
		\hline
		maxP & minP & minDeg & minAvgDeg & batchSM & batchWSM & SWMq & SWMp & $\E{\mu}$ \\
		\hline
			
		24.91 & 26.52 & 26.83 & {\bf 28.21} & 25.81 & 26.92 & 27.68 & 26.65 & 28.71 \\
		22.43 & 23.25 & 24.07 & {\bf 24.99} & 22.69 & 23.69 & 24.41 & 23.93 & 25.51 \\
	  18.64 & 19.03 & 19.11 & {\bf 19.53} & 18.87 & 19.24 & 19.36 & 19.17 & 19.82 \\
		16.36 & 15.57 & 16.62 & {\bf 16.79} & 16.49 & 16.56 & 16.67 & 16.75 & 16.85 \\
		19.51 & 21.36 & 20.69 & {\bf 22.22} & 20.06 & 21.06 & 21.86 & 20.17 & 22.56 \\
		22.74 & 23.20 & 23.53 & {\bf 25.14} & 23.20 & 24.00 & 24.77 & 24.02 & 25.64 \\
		24.16 & 24.38 & 25.02 & {\bf 26.44} & 24.54 & 25.33 & 26.08 & 25.28 & 26.82 \\
		21.36 & 22.07 & 23.18 & {\bf 24.30} & 22.41 & 23.15 & 23.90 & 23.68 & 24.66 \\
		25.39 & 26.06 & 27.20 & {\bf 27.98} & 26.11 & 26.79 & 27.47 & 26.63 & 28.43 \\
		22.87 & 21.56 & 23.61 & {\bf 24.06} & 23.06 & 23.47 & 23.81 & 23.60 & 24.43 \\
		\hline
		21.83 & 22.30 & 22.98 & {\bf 23.97} & 22.32 & 23.02 & 23.60 & 22.99 & 24.34 \\
		\hline
	\end{tabular}
	\end{center}
	\caption{Comparison of different strategies. Each row corresponds to an instance, except the last one which reports the average value of each the strategy over all instances.}
	\label{tab:heu}
\end{table}
		
	Table \ref{tab:heu} shows that all these simple strategies perform very well. Surprisingly, these heuristics are actually close to optimal \emph{clairvoyant} strategies, since $\E{\mu}$ is a valid upper bound for the \emph{non-commit} version of this problem as well. These results indicate that, for the kidney exchange application, the commit requirement in the formulation of the problem is not too restrictive, in that good solutions are still obtainable under this constraint. Notice that strategy {\bf minAvgDeg} in particular outperforms all others in \emph{every} instance of the test set. Moreover, its value stays always relatively close to the upper bound $\E{\mu}$, being on average within $1.5\%$.
		
		 
		A caveat to these results is that the confidence interval of $\pm 0.35$ for the value of the strategies does not allow complete discrimination among the best performing heuristics. The confidence intervals obtained from \eqref{conc1}, due to its generality, seems to be much looser than the actual bounds on the estimates. This is reinforced by Figure \ref{fig:conv}, which displays the sample average as a function of the number of samples; notice that with 2,500 samples the estimate already starts oscillating closely around the reported value of 28.21. We remark that a similar convergence profile holds for the other strategies tested.

\begin{figure}
	\centering
		\includegraphics[scale=0.45]{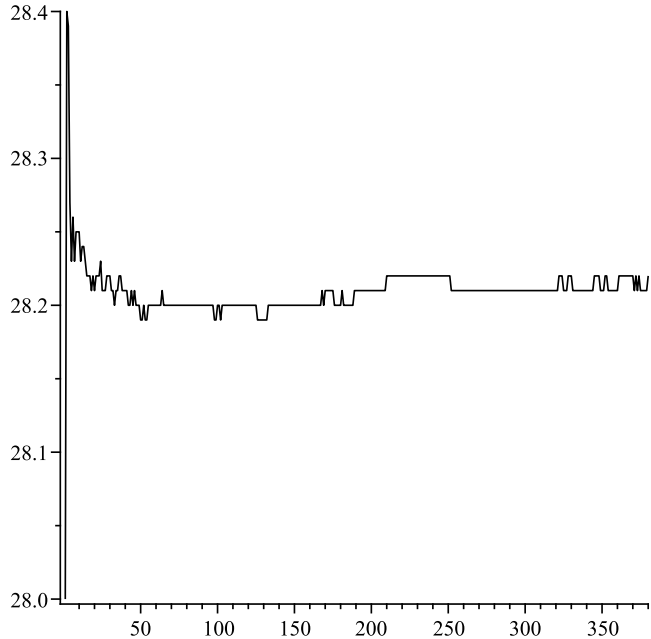}
	\caption{Sample average estimates of the value of {\bf minAvgDeg} on the first instance of the test set as a function of the number of samples. The horizontal axis indicates the number of samples divided by 100.}
	\label{fig:conv}
\end{figure}

		\section{Conclusions and future work} \label{conclusion}
		
		In this paper we considered the query-commit problem, a model for matchings which incorporates uncertainty on the edges of the graph in a way that is suitable for time-sensitive applications. By using the fact that some edges of the graph can be queried without compromising optimality, we show how to obtain an optimal querying strategy for sparse graphs in polynomial time. However, the dependency on the sparsity of the graph is doubly exponential. An interesting open question is to improve this running time, which may indirectly reveal other important properties of optimal strategies. On a similar note, another open question is to prove lower bounds on the computational complexity needed for finding optimal strategies in general graphs. 
	
		In Section \ref{experimental} we evaluated querying strategies over instances of the kidney exchange application and showed that even simple heuristics perform surprisingly well, even when compared to non-committing strategies. An important open question is designing a procedure which is able to provide even better strategies, possibly by starting with a heuristic and successively improving it (e.g. in a local search fashion). However, two hindrances for such procedures are the lack of an algorithm to compute the value of a strategy and the difficulty in representing strategies. As noted previously, the size of a decision tree may be exponential in the size of the input graph. 
	
		Another possibility is to extend the current model to even more realistic setups. For instance, one could consider correlated uncertainty on the edges. In the context of kidney exchanges, the uncertainty on the PRA level of a node introduces correlated uncertainty on all edges incident to it. In addition, recent works in deterministic kidney matchings have considered not only 2-way exchanges but also longer chains of exchanges \cite{david, roth3}, yielding additional transplants. A direction for future research is to study a suitable modification of the query-commit problem which can model uncertainty on longer exchanges. 
		
		Finally, Theorem \ref{thmKidney} indicates the potential of large kidney exchange programs. It would be of great value to obtain a more precise assessment of this potential and to address the logistic problems associated to nationwide transplant programs. 
		
	\paragraph{Acknowledgments.} We thank Tuomas Sandholm, Willem-Jan van Hoeve and David Abraham for helpful discussions.
			

\bibliography{query_commit}
\bibliographystyle{plain}


	\appendix

	\section{Proofs for Section 3}
	
	\subsection{Proof of Lemma \pendantFirst}
	
		Let $S^*$ be an optimal strategy for $G$ and consider the strategy $S$ which first queries $e$ and then proceeds exactly as $S^*$. We clarify what happens in the following situation: the realization $\sigma$ contains both $e$ and another edge $e'$ incident to $e$, and $\Mf{S^*}{G}(\sigma)$ contains $e'$. Clearly $S$ adds $e$ to the matching in the first step and cannot add $e'$ while simulating $S^*$; in this case, $S$ still probes $e'$ whenever $S^*$ does so and adapts accordingly, although $e'$ is not added to the matching. 
		
		We claim that $S$ is optimal, which proves the lemma. The definition of $S$ leads to the following observations. If $e$ does not belong to the realization $\sigma$ then $\Mf{S}{G}(\sigma) = \Mf{S^*}{G}(\sigma)$. On the other hand, if $e$ belongs to $\sigma$ then: (i) the $k$th edge queried by $S$ is exactly the $(k - 1)$th edge queried by $S^*$ (the minus 1 comes from the fact that $S$ queries $e$ before simulating $S^*$) and (ii) every edge added by $S^*$ to the matching $\Mf{S^*}{G}(\sigma)$ is also added to the matching $\Mf{S}{G}(\sigma)$, the only possible exception being a single edge $e'$ of $\Mf{S^*}{G}(\sigma)$ incident to $e$; since $e$ is pendant, $\Mf{S^*}{G}(\sigma)$ contains at most one edge incident to $e$. In the worst case $\Mf{S}{G}(\sigma)$ contains the edges in $\Mf{S^*}{G}(\sigma) \cup \{e\} \setminus \{e'\}$ and we still have $|\Mf{S}{G}(\sigma)| \ge |\Mf{S^*}{G}(\sigma)|$. 
		
		Together, these observations imply that $\EMf{S}{G} \ge \EMf{S^*}{G}$ and hence $S$ is optimal.
	
	
	\subsection{Proof of Lemma 2}
	
	
	
		Let $x_1, x_2, \ldots, x_k$ be the path of $T^S$ induced by its execution over the scenario $\sigma$. By definition of $S$, $a(x_i)$ is pendant in the graph $G^{x_i}$.
		
		By means of contradiction suppose that $\Mf{S}{G}(\sigma)$ is not a maximum matching in $\sigma$, namely there is a matching $M^*$ in $\sigma$ such that $|\Mf{S}{G}(\sigma)| < |M^*|$. Then from Berge's Lemma \cite{lovasz} there must be an augmenting path $P$ in $\sigma$ with respect to $\Mf{S}{G}(\sigma)$ and $M^*$.
		
		Let $a(x_i)$ be an edge of $P$ queried by $T^S$. As mentioned in Section \prelim, $\Mf{T^S}{G}(\sigma) = \{a(x_1), a(x_2), \ldots, a(x_k)\} \cap \sigma$, and since $a(x_i) \in P \subseteq \sigma$ we have that $a(x_i)$ belongs to $\Mf{T^S}{G}(\sigma)$. Then, since $P$ is augmenting, there are two edges $e, e' \in P$ which are incident to $a(x_i)$ and belong to $M^*$. But since $a(x_i)$ is a pendant edge of $G^{x_i}$, it must be that either $e$ or $e'$ does not belong to $G^{x_i}$ and without loss of generality we assume the former. The construction of $G^{x_i}$ implies that there is an edge $a(x_j)$ with $j < i$ such that $a(x_j) \in \sigma$ and $a(x_j)$ is incident to $e$. 
		
		Since $a(x_j) \in \sigma$, we know that $a(x_j)$ also belongs to $\Mf{T^S}{G}(\sigma)$. Since $\Mf{T^S}{G}(\sigma)$ is a matching, it follows that $a(x_j)$ is not incident to an internal node of $P$ and hence it must be incident to the endpoint of $e$ which is also an endpoint of $P$. However, this contradicts the fact that $P$ is an augmenting path, which completes the proof of the lemma.
	
	


	\subsection{Proof of Lemma 3} 

		For each $R \in \mathcal{R}(S(H,e), H)$ let $T^R$ be an optimal decision tree for $R$. Consider the natural CDT $T'$ for $H$ which has $S^{\ro{T'}} = S(H,e)$ and the children of $\ro{T'}$ are the trees $T^R$'s. We claim that $T'$ is optimal for $H$. To argue that, let $T$ be the decision tree corresponding to $T'$. Using the correspondence between CDT's and decision trees, it suffices to prove that $T$ is an optimal decision tree for $H$. 		
		
		We prove that $T_x$ is an optimal decision tree for $H^x$, for every node $x \in T \setminus \bigcup_R T^R$; this is done by reverse induction on the depth of $x$ in $T$. The fact that the trees $T^R$ are optimal removes the necessity of a separate base case, so consider a node $x \in T \setminus \bigcup_R T^R$ and assume that $T_{r(x)}$ is optimal for $H^{r(x)}$ and $T_{l(x)}$ is optimal for $H^{l(x)}$. By the definition of $S$ we have that $a(x)$ is pendant in $H^x$. Therefore, Lemma \pendantFirst asserts that there is an optimal decision tree for $H^x$ whose root queries edge $a(x)$. Then it is easy to see that the optimality of $T_{r(x)}$ and $T_{l(x)}$ implies that actually $T_x$ is one such optimal decision tree for $H^x$, which concludes the inductive step and the proof of the lemma.


	\subsection{Proof of Lemma 4}
	
		First let us relax the definition of $\mathcal{F}$ by defining $\mathcal{F}^+$ as follows. A CDT $T$ for a subgraph $H$ of $G$ belongs to $\mathcal{F}^+$ if each node $x$ of $T$ is either associated to a strategy which consists of querying one edge incident to $V(H^x) \cap V_{\ge 3}$ or it is associated to a strategy $S(H^x, e)$ for some $P_i$ and some $e \in E(H^x) \cap E(P_i)$. Notice $\mathcal{F}$ is the set of CDT's in $\mathcal{F}^+$ which have height at most $9d + 1$.
		
		\begin{claim}
			There is an optimal CDT for $G$ which belongs to $\mathcal{F}^+$.
		\end{claim}
		
		\begin{proof}
			We prove that for each subgraph $H$ of $G$ there is an optimal CDT for $H$ in $\mathcal{F}^+$. We proceed by induction on the number of edges of the subgraph, with trivial base case for subgraphs with no edges. 
			
			Consider a subgraph $H$ of $G$ with at least one edge and let $T^*$ be an optimal querying strategy for it. Suppose that the first edge $e$ queried by $T^*$ is incident to $V_{\ge 3}$. For each $R \in \mathcal{R}(S^{\ro{T^*}}, H)$ let $T^R \in \mathcal{F}^+$ be an optimal CDT for $R$, the existence of which is given by the inductive hypothesis. Then define the decision tree $T$ for $H$ as follows: its root queries edge $e$ and the subtrees of $\ro{T}$ are the trees $\{T^R\}$. Using the recursive equation for $\EMf{T}{H}$ (equation (\ESrecCDT) in the main paper) it is easy to see that the optimality of the trees $\{T^R\}$ implies that $\EMf{T}{H} \ge \EMf{T^*}{H}$ and hence $T$ is optimal. Moreover, $T$ clearly belongs to $\mathcal{F}^+$, which concludes the inductive step in this case.
			
			Now suppose that $e$ is not incident to $V_{\ge 3}$; this implies that $e$ belongs to a path $P_i$. Let $T^*$ be an optimal CDT for $H$ whose root is associated to $S(H, e)$, whose existence is guaranteed by Lemma \mainLemmaSparse. Again, for each $R \in \mathcal{R}(S(H, e), H)$ let $T^R \in \mathcal{F}^+$ be an optimal CDT for $R$. Now we construct the CDT $T$ from $T^*$ by replacing the subtrees of $\ro{T^*}$ by the trees $\{T^R\}$. As in the previous case, the optimality of the trees $\{T^R\}$ implies that $T$ is optimal and also we have that $T \in \mathcal{F}^+$. This concludes the inductive step and the proof of the lemma.
		\end{proof}

		\begin{proof}[Proof of Lemma 4]
			Let $T \in \mathcal{F}^+$ be an optimal CDT for $G$ with the minimum number of nodes. We claim that $T$ has height at most $9d + 1$.
			
			By means of contradiction suppose not and consider a path $Q$ from $\ro{T}$ to one of its leaves which has more than $9d$ internal nodes. Since there are at most $6d$ edges incident to $V_{\ge 3}$ and at most $3d$ paths $P_i$'s in $G$, this means that either: (i) two nodes in $Q$ query the same edge incident to $V_{\ge 3}$ or (ii) two nodes $x, x'$ in $Q$ are associated to two strategies $S(G^x, e)$ and $S(G^{x'}, e')$, where both $e$ and $e'$ belong to the same path $P_i$. Case (i) is forbidden by the definition of a CDT, so we consider case (ii). 
			
			Without loss of generality assume that $x$ is closer to the root of $T$ than $x'$. Notice that by the definition of $S(G^{x'}, e')$ we have that $e' \in E(G^{x'}) \cap E(P_i)$. Now we use the fact that $S(G^x,e)$ removes all edges in $P_i$, that is, for every $R \in \mathcal{R}(S(G^x,e), G^x)$ we have $E(R) \cap E(P_i) = \emptyset$. But the fact that $x'$ is a descendant of $x$ implies that $G^{x'}$ is a subgraph of a residual graph in $\mathcal{R}(S(G^x,e), G^x)$ and hence $E(G^{x'}) \cap E(P_i)$. This contradicts a previous observation that $e' \in E(G^{x'}) \cap E(P_i)$, which implies that $T$ has height at most $9d + 1$ and concludes the proof of the lemma.  
		\end{proof}
		

	\subsection{Proof of Lemma 5}
	
			Consider a tree $T \in \mathcal{F}$; we claim that each node in $T$ has at most 4 children. Equivalently, if $x$ is a node in $T$ we want to upper bound $|\mathcal{R}(S^x, G^x)|$. If $x$ is associated to a strategy $S(G^x, e)$ then as noted previously we have that $|\mathcal{R}(S^x, G^x)| \le 4$. Now if $x$ is associated to a strategy that only queries one edge $e$ incident to $V_{\ge 3}$ then, as in standard decision trees, $\mathcal{R}(S^x, G^x) = \{G^x \setminus N(e), G^x \setminus e\}$. It follows that the outdegree of each tree in $\mathcal{F}$ is at most 4.
		
		Since a tree in $\mathcal{F}$ has height at most $9 d + 1$, the previous degree bound imply that each such tree has at most $4^{9 d + 1}$ nodes. Now notice that each node has one of $e(G)$ possible strategies, because the strategies $S^x$'s allowed in $\mathcal{F}$ are uniquely determined by the choice of an edge. These observations imply that there are at most $e(G)^{4^{9d + 1}}$ trees in $\mathcal{F}$.


	\subsection{Proof of Lemma 6}
	
	Let $P_i = (u_1 e_1 u_2 e_2 \ldots e_q u_{q+1})$ and let $S_a^b$ denote the strategy which queries edges $\{e_k\}_{k = a}^b \cap H$ in order of the indices from $a$ to $b$. That is, if $a < b$ then $S_a^b$ queries the edges $e_a, e_{a + 1}, \ldots, e_b$ and if $a > b$ then it queries edges $e_a, e_{a - 1}, \ldots, e_b$, always ignoring edges outside $H$. To simplify the notation define the random matching $M_a^b \doteq M(S_a^b, H)$ and without ambiguity we use $M_a^b$ to also denote the random variable $|M_a^b|$ corresponding to its cardinality. 
	
	First we prove that we can compute in polynomial time the following quantities corresponding to $S_a^1$, for all $a \ge 1$: $\E{M_a^1|e_1 \in M_a^1}$, $\E{M_a^1|e_1 \notin M_a^1}$ and $Pr(e_1 \in M_a^1)$; then we show how to use this information to compute the values required by the lemma. 
	
	We proceed in a dynamic programming fashion. First, it is trivial to compute the quantities associated to $S_1^1$ and now we want to compute the quantities associated $S_a^1$ assuming that those for the $S_{a'}^1$'s with $a' < a$ have already been computed. It is not difficult to see that $$\E{M_a^1 \Big|e_1 \in M_a^1} = p_{e_a} \left(1 + \E{M_{a - 2}^1 \Big| e_1 \in M_{a - 2}^1}\right) + \left(1 - p_{e_a}\right) \E{M_{a - 1}^1 \Big| e_1 \in M_{a - 1}^1}.$$ Moreover, the analogous expression with complementary conditionings also holds: $$\E{M_a^1 \Big|e_1 \notin M_a^1} = p_{e_a} \left(1 + \E{M_{a - 2}^1 \Big| e_1 \notin M_{a - 2}^1}\right) + \left(1 - p_{e_a}\right) \E{M_{a - 1}^1 \Big| e_1 \notin M_{a - 1}^1}.$$ Finally, we also have $Pr(e_1 \in M_a^1) = p_{e_1} Pr(e_1 \in M_{a - 2}^1) + (1 - p_{e_1}) Pr(e_1 \in M_{a - 1}^1)$. All these expressions can be easily computed using the information about $S_{a-1}^1$ and $S_{a-2}^1$ available by the dynamic programming hypothesis, concluding the proof of our claim. 
	
	By suitably relabeling the edges $e_k$, the above result implies that we can also compute $\E{M_a^q|e_q \in M_a^q}$, $\E{M_a^q|e_q \notin M_a^q}$ and $Pr(e_q \in M_a^q)$ for $a \ge 1$. A final remark is that using the law of total expectation we can also compute $\E{M_a^1}$ and $\E{M_a^q}$.
	
	Now let $a$ be such that $e = e_a$. We show how to compute the desired quantities by the lemma: $\Ep{M}$, $Pr(u_1 \in M)$, $Pr(u_{q + 1} \in M)$ and $Pr(u_1 \in M \wedge u_{q + 1} \in M)$. It is useful to think of $S(H, e_a)$ roughly as the strategy which queries $e_a$ first then queries according to $S_{a'}^1$ and the according to $S_{a''}^q$, for suitable $a'$ and $a''$.	To make this more formal we need to split into a few cases.
	
	\paragraph{Case 1: $1 < a < q$.} Notice that since $e_a \in H$ and $e_a$ is the first edge queried by $S(H, e)$, the event $e_a \in M$ is the same as the event that the edge $e_a$ exists (both which happen with probability $p_{e_a}$). In addition, since by hypothesis $P_i$ is a path with more than one node, we have $u_1 \neq u_{q + 1}$ and thus no edge in the set $\{e_k\}_{k = 1}^{a-1}$ intersects an edge in the set $\{e_k\}_{k = a+1}^{q}$. This guarantees that the outcomes of, say, strategies $S_{a-1}^1$ and $S_{a + 1}^q$ are independent. 
	
	These observations give the following equations:
	\begin{gather*}
		\E{M \Big | e_a \in M} = 1 + \E{M_{a - 2}^1} + \E{M_{a + 2}^q}\\
		\E{M \Big | e_a \notin M} = \E{M_{a - 1}^1} + \E{M_{a + 1}^q}.\\
	\end{gather*}
	
	Also using the fact that $1 < a < q$, and hence $e_a$ is not incident to either $u_1$ or $u_{q+1}$, we obtain that: 
	\begin{gather*}
		Pr(u_1 \in M | e_a \in M) = Pr(u_1 \in M_{a-2}^1) = Pr(e_1 \in M_{a - 2}^1) \\
		Pr(u_{q+1} \in M | e_a \in M) = Pr(u_{q+1} \in M_{a+2}^q) = Pr(e_q \in M_{a+2}^q)\\
		Pr(u_1 \in M | e_a \notin M) = Pr(u_1 \in M_{a-1}^1) = Pr(e_1 \in M_{a - 1}^1)\\
		Pr(u_{q+1} \in M | e_a \notin M) = Pr(u_{q+1} \in M_{a + 1}^q) = Pr(e_q \in M_{a + 1}^q).
	\end{gather*}
	
	Using the additional independence remark made previously, we also have that 
	\begin{gather*}
		Pr(u_1 \in M \wedge u_{q + 1} \in M | e_a \in M) = Pr(e_1 \in M_{a-2}^1) Pr(e_q \in M_{a+2}^q) \\
		Pr(u_1 \in M \wedge u_{q + 1} \in M | e_a \notin M) = Pr(e_1 \in M_{a-1}^1) Pr(e_q \in M_{a+1}^q).
	\end{gather*}		
	
	Therefore, using the fact that $Pr(e_a \in M) = p_{e_a}$ and the laws of total probability and total expectations, we can compute $\Ep{M}, Pr(u_1 \in M), Pr(u_{q+1} \in M)$ and $Pr(u_1 \in M \wedge u_{q + 1} \in M)$ in polynomial time using the information about the $S_a^b$'s. 
	
	\paragraph{Case 2: $a = 1$ or $a = q$.} The equations for the expectations are the same as in Subcase 1.1. Now if $a = 1 \neq q$ then 
	\begin{gather*}
		Pr(u_1 \in M | e_a \in M) = 1 \\
		Pr(u_{q + 1} \in M | e_a \in M) = Pr(u_1 \in M \wedge u_{q + 1} \in M | e_a \in M) = Pr(e_q \in M_{a + 2}^q)\\
		Pr(u_1 \in M |e_a \notin M) = Pr(u_1 \in M \wedge u_{q + 1} \in M | e_a \notin M) = 0 \\
		Pr(u_{q + 1} \in M | e_a \notin M) = Pr(e_q \in M_{a + 1}^q).
	\end{gather*}
		
	Applying a similar reasoning to the case $a = q \neq 1$ we obtain
	\begin{gather*}
		Pr(u_1 \in M | e_a \in M) = Pr(u_1 \in M \wedge u_{q + 1} \in M | e_a \in M) = Pr(e_1 \in M_{a - 2}^1) \\
		Pr(u_{q + 1} \in M | e_a \in M) = 1\\
		Pr(u_1 \in M |e_a \notin M) = Pr(e_1 \in M_{a - 1}^1)\\
		Pr(u_{q + 1} \in M | e_a \notin M) = Pr(u_1 \in M \wedge u_{q + 1} \in M | e_a \notin M) = 0.
	\end{gather*}
	
	Finally, if $a = 1 = q$ then $$Pr(u_1 \in M) = Pr(u_{q + 1} \in M) = Pr(u_1 \in M \wedge u_{q + 1} \in M) = p_{e_a}.$$ Therefore, we can again compute $\Ep{M}, Pr(u_1 \in M), Pr(u_{q+1} \in M)$ and $Pr(u_1 \in M \wedge u_{q + 1} \in M)$ in polynomial time using the information about the $S_a^b$'s. 
	
	\bigskip Since we can calculate the probabilities associated to the $S_a^b$'s in polynomial time and then according to Cases 1 and 2 use this information to calculate directly the values $\Ep{M}$, $Pr(u_1 \in M)$, $Pr(u_{q + 1} \in M)$ and $Pr(u_1 \in M \wedge u_{q + 1} \in M)$, this concludes the proof of the lemma.


	\subsection{Proof of Lemma 7}

		First consider $G \setminus e$ and let $G_1$ and $G_2$ be its connected components (if $G \setminus e$ has only one component then we set $G_2$ to be the empty graph). Since $v(G) = v(G_1) + v(G_2)$ and $e(G) = 1 + e(G_1) + e(G_2)$, we have that $e(G) - v(G) = 1 + (e(G_1) - v(G_1)) + (e(G_2) - v(G_2))$. Since each $G_i$ is connected, $e(G_i) \ge v(G_i) - 1$ for all $i$ and thus $d \ge e(G) - v(G) = e(G_j) - v(G_j)$ for all $j$. This shows that all components of $G \setminus e$ are $d$-sparse. 
		
		Now consider $G \setminus N(e)$ and let $G_1, \ldots, G_k$ be its connected components. Since $G$ is connected it must contain a distinct edge connecting $e$ to each $G_i$. Thus, $e(G) = 1 + \sum_{i = 1}^k (e(G_i) + 1)$. Using the fact that $v(G) = 2 + \sum_{i = 1}^k v(G_i)$, we obtain $$d \ge e(G) - v(G) = -1 + \sum_{i = 1}^k (e(G_i) - v(G_i) + 1).$$ Again using the fact that each $G_i$ is connected, we get $e(G_i) - v(G_i) \ge -1$ for all $i$ and hence $d \ge -1 + (e(G_j) - v(G_j) + 1)$ for all $j$, which shows that all components of $G \setminus N(e)$ are $d$-sparse.
		

\section{Proofs for Section 4}

	\subsection{Proof of Lemma 8}
	
			Let $U = \{u_1, \ldots, u_k\}$ and $V$ be the vertex classes of $G$ and for any node $u\in G$ let $E(u)$ denote the set of edges incident to $u$. Consider the strategy $S$ that queries all edges in $E(u_k)$ (in an arbitrary order), then queries all edges in $E(u_{k - 1})$ and so on. 
			
			We want to upper bound the probability that a node $u_i$ is unmatched in $\Mf{S}{G}$. In order to achieve this, consider the execution of $S$ right before it starts querying edges in $E(u_i)$ and let $M$ denote the random matching of $G$ obtained by $S$ at this point. Suppose $u_i$ is unmatched in $\Mf{S}{G}$, which implies that no edge $(u_i,v)$ could be added to the matching. If $(u_i, v)$ could not be added to the matching then either $(u_i,v)$ does not belong to the realization of $G$ or $v$ is already matched in $M$. Thus, conditioning on the fact that all nodes in $V' \subseteq V(G)$ are unmatched in $M$ we get that 
			\begin{equation}
				Pr\left(u_i \notin \Mf{S}{G} \Big| \bigwedge_{v' \in V'} v' \notin M \right) \le Pr\left(\bigwedge_{v' \in V'} (u_i, v) \notin G \Big | \bigwedge_{v' \in V'} v' \notin M\right). \label{bipartiteEq}
			\end{equation} 
			
			Since $M$ does not depend on the existence of any edge $(u_i,v)$, we can use the independence of the edges in $G$ to obtain $$Pr\left(\bigwedge_{v' \in V'} (u_i, v) \notin G \Big | \bigwedge_{v' \in V'} v' \notin M\right) = Pr\left(\bigwedge_{v' \in V'} (u_i, v) \notin G \right) = \prod_{v' \in V'} Pr((u_i, v) \notin G) \le q^{|V'|}$$ for all $V' \subseteq V(G)$ such that $Pr(\bigwedge_{v' \in V'} v' \notin M) > 0$. Furthermore, notice that in every scenario $M$ leaves at least $i$ nodes from $V$ unmatched, since it can match at most $k - i$ nodes of $U$. Then inequality \eqref{bipartiteEq} reduces to the desired bound $Pr(u_i \notin \Mf{S}{G}) \le q^i$. 
			
			Therefore the probability that $u_i$ is matched in $\Mf{S}{G}$ is at least $1 - q^i$. Since $G$ is bipartite, $\EMf{S}{G}$ is equal to the expected number of nodes of $U$ which are matched, so by linearity of expectation $\EMf{S}{G} \ge \sum_{i = 1}^k (1 - q^i) = k - \sum_{i = 1}^k q^i$. But then for any $\epsilon > 0$ we can bound the last summation as follows:
			\begin{eqnarray*}
				\sum_{i = 1}^k q^i \le \sum_{i = 1}^{\lfloor \epsilon k \rfloor} q^0 + \sum_{i = \lfloor \epsilon k \rfloor}^k q^{\lfloor \epsilon k \rfloor} \le \epsilon k + k q^{\lfloor \epsilon k \rfloor},
			\end{eqnarray*}
		which gives the desired bound of $\EMf{S}{G} \ge (1 - q^{\lfloor \epsilon k \rfloor} - \epsilon) k$.	


\section{Proofs for Section 5}

	\subsection{Proof of Lemma \tightSampling}
			
		Consider the graph $G$ depictured in Figure \ref{fig:tightSampling}. It consists of $n/4$ disjoint paths and every edge has probability 1, except edge $e_1^2$ which has probability $1/2$. Consider the following adaptive strategy $S$: it first queries $e_1^2$ and, in case it belongs to the realization of $G$, $S$ queries edges $\{e^2_i\}$ sequentially in an arbitrary order; in case $e$ does not belong to the realization, $S$ queries edges  $\{e^1_i\}$ and $\{e^3_i\}$ sequentially in an arbitrary order. Notice that in the first case $S$ obtains a matching of size $n/4$, whereas in the second it obtains a matching of size $n/2$. 

\begin{figure}[htbp]
	\centering
		\includegraphics[scale=0.75]{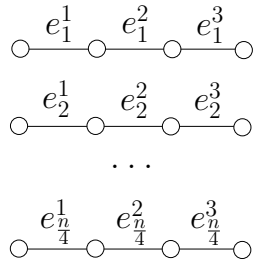}
	\caption{Graph for proof of Lemma \tightSampling.}
	\label{fig:tightSampling}
\end{figure}
			
		Recall that $\sigma_i \sim G$ and that $\eta = (1/k) \sum_{i = 1}^k |\Mf{S}{G}(\sigma_i)|$. From the previous paragraph we know that $|\Mf{S}{G}(\sigma_i)| = n/2 + (n/2) B(1,1/2)$, where $B(a,b)$ denotes a binomial random variable with $a$ trials and success probability $b$. Therefore, $\eta = n/2 + (n/2k) B(k, 1/2)$. Moreover, since $\EMf{S}{G} = \E{\eta} = 3n/4$, we get that
		\begin{equation*}
			Pr\left(\eta \ge \EMf{S}{G} + t\right) = Pr\left(\frac{n}{2} + \frac{n}{2k} B\left(k, 1/2\right) \ge \frac{3n}{4} + t\right) = Pr\left(B\left(k, 1/2\right) \ge \frac{k}{2} + \frac{2k t}{n}\right).
		\end{equation*}
		However, for every $t' \in [0, k/8]$ we can lower bound $Pr(B\left(k, 1/2\right) \ge \frac{k}{2} + t')$ by $\frac{1}{15} \exp(- 16 t^{\prime 2}/k)$ \cite{matousekVondrak}. By our hypothesis on $t$ we have $2k t / n \le k/8$ and hence we can employ this bound on the last displayed inequality, obtaining:
		\begin{equation*}
			Pr\left(\eta \ge \EMf{S}{G} + t\right) \ge \frac{1}{15} \exp\left(- \frac{64 k t^2}{n^2}\right).
		\end{equation*}
		
		Since $\eta$ is symmetric around its expected value $\EMf{S}{G}$, the same upper bound holds for its lower tail:
		\begin{equation*}
			Pr\left(\eta \le \EMf{S}{G} - t\right) \ge \frac{1}{15} \exp\left(- \frac{64 k t^2}{n^2}\right).
		\end{equation*}
		The lemma then follows by combining the bounds on upper and lower the tails.


	\subsection{Proof of Lemma 12}

			Fix $j$ and define the function $g_j(x_1, \ldots, x_{j - 1}, x_{j + 1}, \ldots, x_m)$ as the size of the maximum matching in the subgraph of $G$ which contains edge $e_i$ iff $x_i = 1$; we remark that this subgraph does not contain the edge $e_i$. It is easy to see that for every $x_1, \ldots, x_m \in \{0,1\}$ $$0 \le \mu'(x_1, \ldots, x_m) - g_j(x_1, \ldots, x_{j - 1}, x_{j + 1}, \ldots, x_m) \le 1.$$ 
			
			Now let $M(x_1, \ldots, x_m)$ be a maximum matching in the subgraph of $G$ induced by $x_1, \ldots, x_m$ and recall that $|M(x_1, \ldots, x_m)| = \mu'(x_1, \ldots, x_m)$. Notice that if $\mu'(x_1, \ldots, x_m) > g_j(x_1, \ldots, x_{j - 1}, x_{j + 1}, \ldots, x_m)$ then it must be the case that $e_j$ belongs to $M(x_1, \ldots, x_m)$. Then we can charge the difference between $\mu'$ and the $g_j$'s to the edges in $M(x_1, \ldots, x_m)$:
			\begin{eqnarray*}
				\sum_{j=1}^m \left[\mu'(x_1, \ldots, x_m) - g_j(x_1, \ldots, x_{j - 1}, x_{j + 1}, \ldots, x_m) \right] \le |M(x_1, \ldots, x_m)| = \mu'(x_1, \ldots, x_m)
			\end{eqnarray*}
	which shows that $\mu'(.)$ is self-bounding.


	\section{Concentration inequalities}
	
		For completeness we present two inequalities used to bound large deviations of sums of random variables \cite{BoucheronLB03}.
	
		\begin{lemma}[Hoeffding's inequality] \label{hoeffding}
		Let $X_1, \ldots, X_k$ be independent random variables with $X_i \in [a_i, b_i]$. Let $Y = (1/k) \sum_{i = 1}^k X_i$. Then
		 \begin{equation*}
		 	Pr\left(\left|Y - \E{Y}\right| \ge t\right) \le 2 \exp\left( -\frac{2k^2 t^2}{\sum_{i = 1}^k (b_i - a_i)^2} \right).
		 \end{equation*}
	\end{lemma}

	\begin{lemma}[Bernstein's inequality] \label{bernIneq}
		Let $X_1, \ldots, X_k$ be independent random variables with equal variance $\V{X_i} = \sigma^2$. Let $Y = (1/k) \sum_{i = 1}^k X_i$. Then
		 \begin{equation*}
		 	Pr\left(\left|Y - \E{Y}\right| \ge t\right) \le 2 \exp\left( -\frac{k t^2}{2(\sigma^2 + t/3)} \right).
		 \end{equation*}
	\end{lemma}

\end{document}